\documentclass{article}
\usepackage{arxiv}

\usepackage{hyperref}
\usepackage{microtype}

\usepackage{amssymb}
\usepackage{amsmath}
\usepackage{amsthm}
\usepackage{mathdots}
\usepackage{mathtools}

\usepackage{tensor}

\usepackage{urwchancal}
\DeclareFontFamily{OT1}{pzc}{}
\DeclareFontShape{OT1}{pzc}{m}{it}{<-> s * [1.10] pzcmi7t}{}
\DeclareMathAlphabet{\mathpzc}{OT1}{pzc}{m}{it}

\usepackage{graphicx}
\usepackage{ifpdf}
\ifpdf
\usepackage{epstopdf}
\PrependGraphicsExtensions{.svg}
\fi

\usepackage{xypic}

\newtheorem{theorem}{Theorem}[section]
\newtheorem{lemma}[theorem]{Lemma}

\newtheorem{remark}[theorem]{Remark}

\providecommand{\R}{\mathbb{R}}

\providecommand{\SE}{\mathbf{SE}}

\providecommand{\calM}{\mathcal{M}}

\providecommand{\td}{\mathrm{d}}

\providecommand{\ddt}{\frac{\td}{\td t}}

\providecommand{\dtau}{\td \tau}

\usepackage{accents}
\makeatletter
\providecommand{\scirc}{%
    \hbox{\fontfamily{\rmdefault}\fontsize{0.4\dimexpr(\f@size pt)}{0}\selectfont{\raisebox{-0.52ex}[0ex][-0.52ex]{$\circ$}}}}

\makeatother

\mathchardef\mhyphen="2D

\providecommand{\etal}{\textit{et al.}~}

\newcommand{\pb}[2]{\left\langle #1 \;\vline\; #2 \right\rangle}
\usepackage[all]{xy}
\newdir{/|}{%
	@{|}}
\newdir_{/|}{%
	@{|}}
\newdir^{/|}{%
	@{|}}
\newdir{/|<}{%
	@{|}*@{<}}
\newdir_{>/|}{%
	@_{>}*@{|}}
\newdir^{>/|}{%
	@^{>}*@{|}}
\newdir{/|<}{%
	@{|}*@{<}}
\newdir_{/|<}{%
	@{|}*@_{<}}
\newdir^{/|<}{%
	@{|}*@^{<}}

\newlength{\bgwidth}
\newlength{\bgheight}

\usepackage{tikz} 

\newcommand{\bgel}[1]{
	\settowidth{\bgwidth}{#1}
	\settowidth{\bgheight}{#1}
	\parbox[c]{\bgwidth}{#1}
}

\providecommand{\margin}[1]{} 
\providecommand{\todo}[1]{}
\providecommand{\aside}[1]{} 

\newcommand{\papercitation}{
Mahony, R (2019), `` A novel passivity-based trajectory tracking control for conservative mechanical systems'', \emph{Proceedings of the IEEE Conference on Decision and Control}, pp.~4259-–4266.
}

\newcommand{\publicationdetails}
{
\papercitation \\
\copyrightNoticeIEEE{2020} 
}
\newcommand{\publicationversion}
{Author accepted version}

\begin{document}

\title{A Novel Passivity-Based Trajectory Tracking Control For Conservative Mechanical Systems}
\headertitle{Passivity-Based Trajectory Tracking Control}

\author{
\href{https://orcid.org/0000-0002-7803-2868}{\includegraphics[scale=0.06]{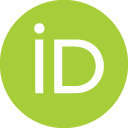}\hspace{1mm}
Robert Mahony}
\\
	Systems Theory and Robotics Group \\
	Australian National University \\
    ACT, 2601, Australia \\
	\texttt{Robert.Mahony@anu.edu.au} \\
}

\maketitle

\begin{abstract}
Most passivity based trajectory tracking algorithms for mechanical systems can only stabilise reference trajectories that have constant energy.
This paper overcomes this limitation by deriving a single variable Hamiltonian model for the reference trajectory and solving along the constrained trajectory to obtain a \emph{reference potential}.
This potential is then used as the model to shape the energy of the true system such that its free solutions include the desired reference trajectory.
The proposed trajectory tracking algorithm interconnects the reference and true systems through a virtual spring damper along with an outer-loop energy pump/damper that stabilises the desired energy level of the interconnected system, ensuring asymptotic tracking of the desired trajectory.
The resulting algorithm is a fully energy based trajectory tracking control for non-stationary trajectories of conservative mechanical systems.
\end{abstract}



\section{Introduction}
\label{sec:intro}

The industry standard trajectory tracking control algorithms for mechanical systems depend on the well known Euler-Lagrange formulation of the dynamics in terms of generalised coordinates $q \in \R^n$ (for example \cite{SpoHutVid_2006})
\begin{align}
M(q) \ddot{q} = - C(q,\dot{q}) \dot{q} - \frac{\partial U}{\partial q}(q) + \tau,
\label{eq:q_dyn}
\end{align}
for a bounded generalised inertia matrix $M(q) > 0$, Coriolis matrix $C(q,\dot{q})$, potential $U(q)$, and exogenous generalised input $\tau$.
The goal is to find an input $\tau(t)$ such that the systems trajectory $q(t)$ converges asymptotically to a known reference trajectory $q_{\text{r}}(t)$.
The standard approach (computed torque with passive damping injection) exploits the Euclidean structure of the generalised coordinates $q \in \R^n$ explicitly in creating an error $\tilde{q}= q - q_{\text{r}}$.
A trajectory and state dependent feed-forward torque $\tau_{\text{ff}}$ is computed to bring the error dynamics into a time-varying linear structure
\begin{align}
M(q) \ddot{\tilde{q}} = - C(q,\dot{q})\dot{\tilde{q}}  + \tilde{\tau}
\label{eq:q_tilde_error_dyn}
\end{align}
where $\tilde{\tau} = \tau - \tau_{\text{ff}}$ is the residue torque.
From here the algebraic passivity properties of $M(q)$ and $C(q,\dot{q})$ are exploited to show a $PD$ control (from $\tilde{q}$ to $\tilde{\tau}$) stabilizes the error coordinates $\tilde{q} \rightarrow 0$ and hence $q(t) \rightarrow q_\text{r}(t)$.

Although this approach has been an effective tool in nonlinear control design for many years, it is subject to funamental limitations:
\begin{enumerate}
\item The error $\tilde{q}$ is a Euclidean difference that has no physical  interpretation in terms of the structure of the mechanical system and is only defined locally (in a coordinate patch) for a general system with a manifold state-space.
	
	\item The feed-forward computed torque $\tau_{\text{ff}}$ has no clear physical interpretation.
Computing this requires real-time measurements of $(q, \dot{q})$ as well as $(q_\text{r},\dot{q}_\text{r},\ddot{q}_\text{r})$ and a good model of the system.
\end{enumerate}

As a consequence, systems for which the state space is inherently non-Euclidean, such as mobile robots moving in $\SE(3)$, or systems where the model is uncertain, such as manipulators moving objects that are not well characterised, are difficult to address using the classical tracking framework.
Furthermore, due to the complex nonlinear dependence of the feed-forward control on the system state many of the tricks used for velocity free passivity based control, or passivity based robust disturbance rejection \cite{2007_Secchi_book} cannot be exploited.

In this paper, I propose a `true' passivity based tracking control for conservative mechanical systems.
The key innovation is to use the reference trajectory as a parametrization for a single coordinate Hamiltonian system with kinetic and potential energy inherited from the full system.
The desired reference trajectory behaviour corresponds to constant velocity motion of the single parameter reference system and the exogenous input that generates this trajectory is uniquely defined.
Due to the fact that the system has a single coordinate the input associated to the reference trajectory can be integrated to generate a \emph{reference potential}.
Interpreting the reference input as shaping the potential of the reference system, the constant velocity solution becomes a `free' solution of the single coordinate Hamiltonian reference system.

The second contribution of the paper is to show how to use the reference potential to shape the true system potential energy in such a way that the reference trajectory is a free solution of the (shaped) full Hamiltonian system.
The proposed control then interconnects the shaped reference and shaped true system using non-linear spring dampers in order that the response of both systems is synchronised.
This step is the primary control that forces the true system to converge towards the desired reference behaviour.
The final stage of the control introduces an energy pump that stabilises the energy level in both true and reference system to the desired reference energy level ensuring that the response of the interconnected systems converges asymptotically to the desired reference trajectory.

The resulting control is inherently energy based and does not involve linearisation of the dynamics.
The response should be highly robust to model errors and external disturbances and particularly well suited to trajectory tracking in the presence of unknown but energy bounded disturbances.
The philosophy of the design draws strongly from the inspiration of Full \etal \cite{1999_Koditschek_JEB} and Westervelt \etal \cite{2003_Grizzle_2003} in stabilising an embedded behaviour in a larger mechanical system to obtain emergent behaviour in the zero dynamics.
A limitation of the approach is that the desired trajectory $q_r(t)$ is only tracked up to an arbitrary time-offset $t_0$; that is, $q(t) \rightarrow q_r(t - t_0)$.

\section{Literature Review}
\label{sec:literature}

Trajectory tracking control for dynamical mechanical systems is a classical problem robotics \cite{SloLi91,SpoHutVid_2006}.
The industry standard method for trajectory tracking is that of computed torque control, first developed in the seventies \cite{1973_Markewitz_techrep,1982_Paul_MITpress}.
The summary of this approach covered in the introduction of the paper was based on classical robotics texts \cite{Cra89,SloLi91,SpoHutVid_2006}.
Set point regulation has received considerably more attention.
In 1981 Takegaki \etal \cite{1981_Takegaki_JDSMC} demonstrated that a PD control with gravity compensation can regulate a robotic manipulator to any fixed point in its workspace, and this work was generalised to adaptive passivity based set point regulation during the eighties
\cite{1986_Middleton_cdc,1989_Weiping_scl} with the global stability analysis for PD, PID, and PI$^2$D controls provided in the nineties \cite{1995_Kelly_Robotica,1998_Ortega,2000_Lefeber_PhD}.
Although approximate trajectory tracking can be accomplished with joint based PD \cite{1987_Koditschek_techrep,1988_Kawamura_ICRA} most theoretical developments of these techniques have been aimed at regulation of set-point or constant energy trajectories \cite{2000_Lefeber_PhD}.

A seminal contribution to tracking trajectory control is the work by Fujimoto \etal \cite{2003_Fujimoto_Automatica} in 2003.
In this work the authors state; ``The key idea to realize [the objective of tracking control] is to embed the desired trajectory into the Hamiltonian function of the original system.''
They provide a characterisation of a generalised canonical transform that transforms a port Hamiltonian system into one written in terms of the error coordinates
\begin{align}
M(\tilde{q}) \ddot{\tilde{q}} = - C(\tilde{q},\dot{\tilde{q}}) \dot{\tilde{q}} - \frac{\partial U}{\partial \tilde{q}}(\tilde{q}) + \tau,
\label{eq:q_dyn}
\end{align}
(compare to \eqref{eq:q_tilde_error_dyn}) which can be stabilised using passive damping injection.
The approach is conceptually rigorous but formally requires the solution of a Partial Differential Equation (PDE) to find the canonical transformation that defines the global error $\tilde{q}$ coordinates.
This requirement is reminiscent of the IDA-PCA algorithms \cite{1998_Ortega} developed around the same time and the complexity of the design process has limited their adoption in the wider community.

A separate research thread I draw from in this paper is summarised in the philosophy stated in Westervelt \etal \cite{2003_Grizzle_2003} ``to encode the dynamic task via a lower dimensional target, itself represented by a set of differential equations'' \cite{2000_Nakanishi_TRO} motivated by the study of how biological systems move \cite{1999_Koditschek_JEB}.
This work has been important in development of advanced trajectory tracking algorithms for gait control for humanoid robots \cite{2003_Grizzle_2003} in particular, and in general for control of robotic motion.
I also draw from the pioneering work by Spong \cite{1995_Spong_CSM} on energy pumping for stabilisation of constant energy set-points for under-actuated mechanical systems.

\section{Problem formulation}
\label{sec:prob}

Let $q \in \R^N$ denote (local) generalized coordinates for a fully actuated mechanical system with generalized inertia matrix $M(q) > 0$ and potential function $U(q)$.
Lagrange's equations yield dynamics \eqref{eq:q_dyn} for a Coriolis matrix $C(q,\dot{q})$.
The passivity of the system is expressed algebraically by the relationship that
\[
\left( \frac{1}{2}\frac{\td}{\td t} M(q) + C(q,\dot{q}) \right)^\top =
- \left( \frac{1}{2}\frac{\td}{\td t} M(q) + C(q,\dot{q})\right)
\]
is skew symmetric along trajectories of the system.
The Hamiltonian for this system is
\begin{align}
H(q,\dot{q}) = \frac{1}{2} \dot{q}^\top M(q) \dot{q} +  U(q)
\label{eq:H}
\end{align}
and passivity of the system ensures that
\begin{align}
\ddt H(q, \dot{q}) = \pb{\tau}{\dot{q}}
\label{eq:ddt_H}
\end{align}
where $\pb{\cdot}{\cdot}$ is notation for the power bracket between a generalised input and a velocity, and can be read as a simple inner product $\tau^\top \dot{q}$ in local coordinates.

The trajectory tracking problem is associated with a pre-specified reference trajectory $q_{\text{r}}(t)$.
In this paper, I will consider only non-stationary reference trajectories that are bounded, twice differentiable with bounded differentials.
The non-stationary assumption,
\begin{align}
\dot{q}_\text{r}(t) \not= 0,
\label{eq:non-stationary}
\end{align}
is the only unusual assumption.
Indeed, stationary trajectories are exactly those trajectories that can be easily tackled using classical passivity based techniques \cite{1998_Ortega}.
This assumption emphasises the fundamental differences between the present and previous energy based approaches to tracking.
Although I believe that it may be possible to relax this condition in future work, the assumption considerably simplifies the present work.
In fact, I will assume slightly more than \eqref{eq:non-stationary}, that the velocity is bounded away from zero.
That is, that there exists a uniform bound $\beta > 0$ such that $\dot{q}_\text{r}^\top M(q_\text{r}) \dot{q}_\text{r} > \beta$ for all time.

The nominal goal of the trajectory tracking problem is to find a control signal $\tau$ for \eqref{eq:q_dyn} depending on the known reference trajectory $q_\text{r}$ and its derivatives $(\dot{q}_\text{r},\ddot{q}_\text{r})$ along with state measurements $(q, \dot{q})$ such that for any initial condition $q_0$ the solution $q(t;q_0)$ of the closed-loop system \eqref{eq:q_dyn} converges to the reference, $q(t;q_0) \rightarrow q_r(t)$ asymptotically.
The proposed algorithm achieves $q(t;q_0) \rightarrow q_r(t - t_0)$ asymptotically for some time-offset $t_0 \in \R$.
Although this is not explicitly a solution to the classical trajectory tracking problem it is clearly applicable to a wide range of applied problems and the difference again emphasises the difference in approach.

\section{Reference trajectory}
\label{sec:ref_traj}

In this section, I consider the dynamics of the reference trajectory.
The approach taken is to reparametrize the reference trajectory as a one parameter mechanical system constrained to follow a \emph{reference path}.
The key advantage of this reformulation is to derive a potential function that encodes the information necessary to generate the exogeneous generalised force input associated with the reference trajectory.
The goal of the section is to formulate a closed Hamiltonian system for which the reference trajectory is a free solution.

The starting position is to consider the reference $q_{\text{r}}(t)$ as a solution of the mechanical system \eqref{eq:q_dyn}.
In particular, there must be a reference input $\tau_{\text{r}}$ associated with the actual generalised forces required to track the reference trajectory.
More correctly, we consider a reference trajectory to be a quadruple of variable trajectories $(q_{\text{r}}(t), \dot{q}_{\text{r}}(t), \ddot{q}_{\text{r}}(t), \tau_{\text{r}}(t))$ defined on a time interval $[0, \infty )$ that satisfy the system relationship
\begin{equation}\label{eq:reference_trajectory}
M(q_{\text{r}}) \ddot{q}_{\text{r}} = - C(q_{\text{r}},\dot{q}_{\text{r}})\dot{q}_{\text{r}} - \frac{\partial}{\partial q} U(q_{\text{r}}) + \tau_{\text{r}}
\end{equation}
for $t \in (0, \infty)$.
Note that the reference trajectory information, including the reference input $\tau_{\text{r}}$ can be computed in advance.

Let $s$ be a path parameter for the reference trajectory.
That is we consider the map $q_{\text{r}} : \R \rightarrow \R^n$, $s \mapsto q_{\text{r}}(s)$.
Consider the mechanical system obtained by taking $s$ to be a generalised coordinate $s :=s(t)$ for a new one-dimensional mechanical system that characterises the motion of the reference along its predefined trajectory.
That is, as $s(t)$ varies in one dimension, the reference variable $q_\text{r}$ moves along the constrained \emph{path} $q_{\text{r}}(s(t))$.
Define the partial differentials of $q_{\text{r}}$ with respect to $s$ by
\[
q'_{\text{r}}(s) := \frac{\partial q_{\text{r}}}{\partial s}, \quad
q''_{\text{r}}(s) := \frac{\partial^2 q_{\text{r}}}{\partial s^2}
\]
For a solution $s := s(t)$ one has $\dot{q}_{\text{r}} = q'_{\text{r}}(s) \dot{s}$ and
\[
\ddot{q}_{\text{r}} = q'_{\text{r}}(s) \ddot{s} + q''_{\text{r}}(s) \dot{s}^2.
\]
The kinetic energy of the constrained reference system is
\begin{align}
	K_{\text{r}}(s, \dot{s}) = \frac{1}{2} q'_{\text{r}}(s)^\top M(q(s)) q'_{\text{r}}(s) \dot{s}^2  =: \frac{1}{2} M_{\text{r}}(s) \dot{s}^2
	\label{eq:kinetic_s}
\end{align}
which defines $M_{\text{r}}(s)$ while the potential is a simple reparameterization $U_{\text{r}}(s) = U(q_{\text{r}}(s))$.

The Lagrangian for this system is given by $\mathcal{L}_{\text{r}}(s,\dot{s}) = K_{\text{r}}(s,\dot{s}) - U_{\text{r}}(s)$.
Consider the Euler-Lagrange equations $\frac{\td}{\td t} \frac{\partial }{\partial \dot{s}} \mathcal{L} - \frac{\partial}{\partial s} \mathcal{L} = e$ where $e$ is a scalar exogenous input added as a driving term.
It is easily verified that
\begin{align}
M_{\text{r}}(s) \ddot{s} = - \frac{1}{2} \frac{\partial}{\partial s}M_{\text{r}}(s) \dot{s}^2 - \frac{\partial}{\partial s} U_{\text{r}}(s) +  e.
\label{eq:EulerLagrange_s}
\end{align}

The dynamics \eqref{eq:EulerLagrange_s}, for a general input $e$, models the motion of a mechanical single variable system along the parametrization provided by $q_{\text{r}}(s)$.
The reference trajectory itself $q_{\text{r}}(t)$ is recovered by the particular solution $s(t) = t$ and will correspond to a particular choice of scalar reference input $e = e_{\text{r}}(s(t))$.
Since the defining equations are time-invariant, then any solution $s(t) = t - t_0$, $t_0 \in \R$, will also be a solution \eqref{eq:EulerLagrange_s} associated with the offset reference trajectory $q_{\text{r}}(t - t_0)$ for $t-t_0 \in (0,\infty)$.
It follows that the scalar reference input $e_{\text{r}}(s)$ can be computed as function of path parameter $s$ independent of time.
Substituting the relationships $\dot{s} = 1$, $\ddot{s}= 0$, characteristic of the trajectory $s(t) = t - t_0$, into \eqref{eq:EulerLagrange_s} yields
\begin{align}
e_{\text{r}}(s)
& = \frac{\partial}{\partial s} \left( U_{\text{r}}(s) + \frac{1}{2}M_{\text{r}}(s)  \right).
\label{eq:e_r}
\end{align}

Define a \emph{reference potential}
\[
W_{\text{r}}(s) = - U_\text{r}(s) - \frac{1}{2}M_{\text{r}}(s)
\]
The potential $W_{\text{r}}(s)$ encodes information on the scalar reference input associated with the reference trajectory for the scalar system \eqref{eq:EulerLagrange_s}.
The input $e_{\text{r}}(s) = -\frac{\partial}{\partial s} W_{\text{r}}(s)$ \eqref{eq:e_r} should be thought of as a potential shaping input where the term $-\frac{\partial}{\partial s} U_{\text{r}}(s)$  cancels the effect of the potential $U_{\text{r}}$ inherited from the full system, while the term $-\frac{\partial}{\partial s} \frac{1}{2}M_{\text{r}}(s)$
is associated with a new shaped potential that compensates the energy balance for variation in the kinetic energy $\frac{1}{2}M_{\text{r}}(s)\dot{s}^2$ along the reference trajectory.
That is, the reference potential $W_\text{r}$ shapes the potential of the reference system to supply the power required to track the reference trajectory.

\begin{lemma}\label{lem:dotW_is_poweralongref}
For any $t_0 \in \R$ let $q_\text{r}(t - t_0)$ be the reference trajectory offset by time $t_0$ and fix $s(t) = t - t_0$.
Then
\begin{align}
\ddt W_\text{r}(s(t)) = - \pb{\dot{q}_\text{r}}{\tau_\text{r}}
\end{align}
\end{lemma}

\begin{proof}
Recalling \eqref{eq:ddt_H} then by construction
\[
\ddt H_\text{r}(s,\dot{s}) = \ddt H(q_\text{r}, \dot{q}_\text{r}) = \pb{\dot{q}_\text{r}}{\tau_\text{r}}
\]
Note that
\begin{align*}
H(q_\text{r},\dot{q}_\text{r}) & = \frac{1}{2} \dot{q}_\text{r}^\top M(q_\text{r}) \dot{q}_\text{r} + U(q_\text{r}) \\
& = \frac{1}{2} q'_\text{r}(s)^\top M(q_\text{r}(s)) q_\text{r}(s)'\dot{s}(t)^2 + U_\text{r}(s(t)) \\
& = \frac{1}{2} M_\text{r}(s) + U(s) = - W_\text{r}(s)
\end{align*}
where the last line follows since $\dot{s}(t) = 1$ for the solution considered.
Differentiating this relationship one has
\begin{align*}
\ddt  W_\text{r}(s) & = - \ddt H(q_\text{r},\dot{q}_\text{r})
=
- \pb{\dot{q}_\text{r}}{\tau_\text{r}}.
\end{align*}
\end{proof}

Set $e = e_{\text{r}}(s) + e_\text{c}$ where $e_\text{c}$ is a coupling input from the reference system to the true system and will be defined in \S\ref{sec:control}.
The Euler-Lagrange equations for \eqref{eq:EulerLagrange_s} become
\begin{align}
M_{\text{r}}(s) \ddot{s}
& = - \frac{1}{2} \frac{\partial M_{\text{r}}}{\partial s}(s) \dot{s}^2
- \frac{\partial U_\text{r}}{\partial s}(s) - \frac{\partial}{\partial s}  W_{\text{r}}(s) + e_\text{c} \notag \\
& = - \frac{1}{2} \frac{\partial M_{\text{r}}}{\partial s}(s) \dot{s}^2
+ \frac{1}{2} \frac{\partial M_{\text{r}}}{\partial s}(s) + e_\text{c} \notag \\
& = - \frac{1}{2} \frac{\partial M_{\text{r}}}{\partial s}(\dot{s}^2  - 1) + e_\text{c}
\label{eq:s_dyn}
\end{align}
The dynamics \eqref{eq:s_dyn} are associated with the port Hamiltonian system with Hamiltonian
\begin{align}
H_{\text{r}}(s,\dot{s}) = \frac{1}{2}M_{\text{r}}(s) \dot{s}^2 + U_\text{r}(s) + W_{\text{r}}(s).
\label{eq:H_r}
\end{align}
One has
\begin{align*}
\ddt H_{\text{r}}(s,\dot{s}) & = M_{\text{r}}(s) \ddot{s}\dot{s} +
\frac{1}{2} \frac{\partial M_{\text{r}}}{\partial s}(s) \dot{s}^3
- \frac{1}{2} \frac{\partial M_{\text{r}}}{\partial s}(s) \dot{s}\\
& = -  \frac{1}{2} \frac{\partial M_{\text{r}}}{\partial s}(\dot{s}^2  - 1)\dot{s}  + e_\text{c} \dot{s}
+
\frac{1}{2} \frac{\partial M_{\text{r}}}{\partial s}(s) \dot{s}^3 \\
& \quad\quad\quad\quad - \frac{1}{2} \frac{\partial M_{\text{r}}}{\partial s}(s) \dot{s} \\
& = \pb{\dot{s}}{e_\text{c}}
\end{align*}
where $\pb{\dot{s}}{e_\text{c}} = \dot{s}e_\text{c}$ is the power port for the scalar reference system.

It is easily verified that the unforced or zero dynamics of \eqref{eq:s_dyn} (with $e_\text{c} \equiv 0$) include the solution $\ddot{s} = 0$, $\dot{s}= 1$ as expected.
This solution is associated with a constant (shaped) energy level
\begin{align}
E_\text{r}(s, \dot{s})
& =
H_\text{r}(s, \dot{s}) + W_\text{r}(s) \notag \\
& =
\frac{1}{2} M_\text{r}(s) \dot{s}^2 + U_\text{r}(s) - U_\text{r}(s) - \frac{1}{2}M_\text{r}(s) \notag  \\
& =
\frac{1}{2} M_\text{r}(s) - \frac{1}{2}M_\text{r}(s) = 0
\label{eq:E_r_zero}
\end{align}
since $\dot{s} \equiv 1$ along the reference trajectory.
This energy level is, however, not a minimum energy level for the system, with trajectories in phase portrait (Fig.~\ref{fig:phase_portrait}) that lie inside the $\dot{s} = \pm1$ lines all corresponding to lower energy levels.
The minimum energy trajectories for the reference system will correspond to the stationary points $\dot{s}= 0$.
These stationary points occur at points where $\frac{\partial M_{\text{r}}}{\partial s} = 0$ and their stability will depend on the sign of
$\frac{\partial^2 M_{\text{r}}}{\partial s^2} = 0$ (see Fig.~\ref{fig:phase_portrait}).
The variation in $M_\text{r}(s)$ is a function of the reference trajectory, both in the change in magnitude of $q'_{\text{r}}(t)$ as well as the distribution of inertia in $M(q(s))$.
Thus, reference trajectories with high accelerations and complex inertia matrices will lead to complex zero dynamics in the reference system, while tracking simply trajectories with near constant velocities and homogeneous inertia will lead to simple zero dynamics.
Interestingly, the unforced system also includes solutions associated with negative time evolution of the reference system, and a zero energy solution $\dot{s} = -1$ in particular.
\margin{Check this.}

\begin{figure}  
\vspace{2mm}
	\begin{center}
		\includegraphics[width=8.4cm]{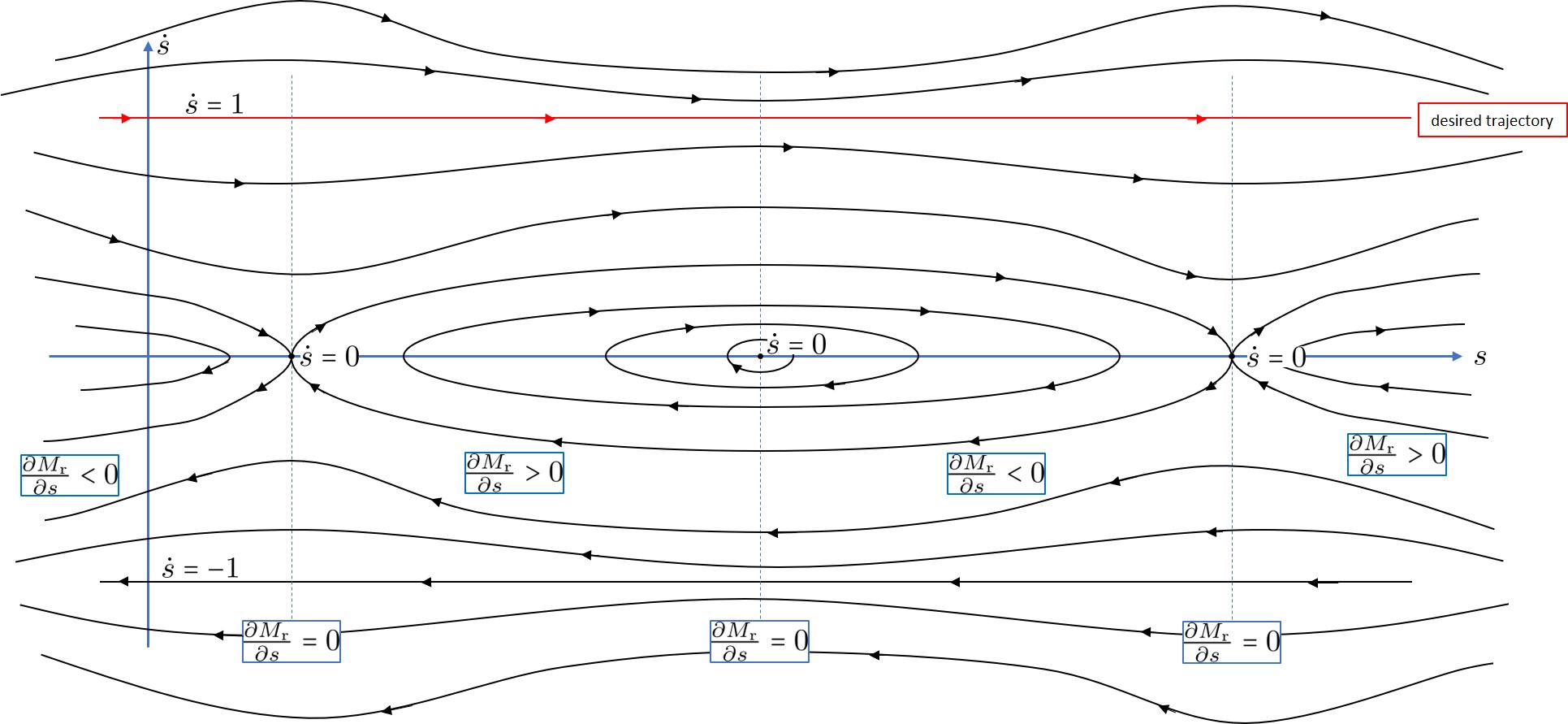}    
		\caption{Phase portrait of the zero (free Hamiltonian) dynamics of \eqref{eq:s_dyn}.
The trajectory shown in red is the desired reference trajectory characterised by $\dot{s} = 1$, $\ddot{s} = 0$.
The phase portrait of the zero dynamics (output $\dot{E} \equiv 0$) is strongly dependent on the sign of $\frac{\partial M_\text{r}}{\partial s}(s)$.
}
		\label{fig:phase_portrait}
	\end{center}
\end{figure}

\section{Shaping, Interconnection and Synchronisation}
\label{sec:control}

In this section I present the first part of the proposed tracking control design.
The idea is to shape the potential of the true system with the reference potential $W_\text{r}$ and then couple the true and reference systems together using passive interconnections.
The resulting control synchronises the state $q(t)$ of the true system \eqref{eq:q_dyn} with the state $q_\text{r}(s(t))$ of the reference system \eqref{eq:s_dyn}.
The final closed-loop system will require an additional energy pumping control discussed in \S\ref{sec:energy_stabilisation} to stabilise the desired energy level set.

The principle control used for synchronization between the true and reference systems is a spring potential introduced between $q(t)$ and $q_\text{r}(s(t))$.
Let $\phi : \calM \times \calM \rightarrow \R$ be a positive definite smooth cost function with $\phi(q, q_\text{r}) = 0$ if and only if $q = q_\text{r}$.
Note that there is no need to compute the linear error $q - q_\text{r}$ in order to find such a spring potential function although the potential $\frac{1}{2}K |q - q_\text{r}|^2$ used in the classical tracking algorithms is an example of such a potential.

The generalised inertia $M(q)$ provides a Riemannian metric on the tangent bundle $T \calM$.
Using this metric one can measure the \emph{speed} of a trajectory $\dot{q}$ by computing $\sqrt{\dot{q}(t)^\top M(q(t)) \dot{q}(t)}$.
In a similar manner one can measure the desired speed of the reference trajectory at a point $s$ by computing
$\sqrt{q'_\text{r}(s)^\top M(q(s)) q'_\text{r}(s)} = \sqrt{M_\text{r}(s)}$.
Note that the reference speed is associated with the desired reference trajectory and does not depend on the velocity of the path parameter $\dot{s}$.
Since only non-stationary reference trajectories were considered, the reference speed $\sqrt{M_\text{r}(s)} \not= 0$.
The ratio of these two speeds as a function of time
\begin{align}
\dot{\sigma}(t) := \frac{\sqrt{\dot{q}(t)^\top M(q(t)) \dot{q}(t)}}
{\sqrt{M_\text{r}(s(t))}}
\label{eq:sigma_dot}
\end{align}
can be interpreted as a measure of the speed at which the true system is moving relative to a desired reference speed associated with a point on the reference trajectory specified by the path parameter $s(t)$.
If $q(t) \approx q_\text{r}(s(t))$ and $\dot{q}(t) \approx q_\text{r}'(s(t))$ then $\dot{\sigma} \approx 1$, that is the true trajectory is moving at the same speed as the desired trajectory.

The integral of \eqref{eq:sigma_dot},
\begin{align}
  \sigma(t) = \int_{0}^t \dot{\sigma}(\tau) \dtau,
\label{eq:sigma_integral}
\end{align}
can be interpreted as a (relative) path parameter for the true system, analogous to the path parameter $s(t)$ introduced for the reference trajectory $q_\text{r}$.
If $q(t) \approx q_\text{r}(s(t))$ and $\dot{q}(t) \approx q_\text{r}'(s(t))$ then $\sigma(t) \approx t - t_0$.
The new path parameter $\sigma$ will be used as a storage variable in a copy of the reference potential $W_\text{r}(\sigma)$ that can be used to shape the potential for the true system.

Define $\tau_{\dot{q}} \in T^\star_q \calM$, the dual (generalised) torque to a velocity $\dot{q} \in T_q \calM$, to be the unique element that solves
\begin{align}
\pb{\tau_{\dot{q}}}{V} = \dot{q}^\top M(q) V
\label{eq:tau_dot_q}
\end{align}
for all $V \in T_q \calM$.
Note that $\tau_{\dot{q}} = 0$ in the case where the velocity $\dot{q} = 0$ is zero.
Recalling the reference input $\tau_\text{r}$ associated with the reference trajectory \eqref{eq:reference_trajectory}, define the \emph{constraint  force}
\begin{align}
\tau^0_\text{r}
& := \frac{\pb{\tau_{\dot{q}}}{\dot{q}}}{M_\text{r}(s)} \tau_r -
 \frac{\pb{\dot{q}}{\tau_r}}{M_\text{r}(s)}\tau_{\dot{q}}
\label{eq:tau_r^0}
\end{align}
By construction $\pb{\tau_\text{r}^0}{\dot{q}} = 0$ and thus does no work on the system.
In the asymptotic case, where $\sigma(t) = s(t) = t - t_0$ then $\dot{\sigma}^2 = 1$ and $\frac{\dot{q}^\top M(q) \dot{q}}{M_\text{r}(s)} = \frac{\pb{\tau_{\dot{q}}}{\dot{q}}}{M_\text{r}(s)} = 1$.
It follows that in this case $\tau^0_\text{r} = \tau_\text{r} -
\frac{\pb{\dot{q}}{\tau_r}}{M_\text{r}(s)}\tau_{\dot{q}}$ is the component of the reference input $\tau_\text{r}$ orthogonal to the dual torque $\tau_{\dot{q}}$.
That is, the component of the reference input that acts orthogonal to the velocity of the system as a constraint force and steers the system to follow the desired trajectory.
The power required to accelerate and decelerate the system along the trajectory is not provided by the constraint force $\tau^0_\text{r}$ and must be separately modelled.
The active force in direction $\tau_{\dot{q}}$ will be coupled to the potential $W_\text{r}(\sigma)$ \eqref{eq:mu}.

A key requirement of the proposed control, separate from the synchronisation of $q(t)$ and $q_\text{r}(s(t))$, is to synchronise the two path parameters $\sigma(t)$ and $s(t)$.
To this end define a second potential
\[
\psi(s,\sigma) := \frac{\kappa}{2} |s - \sigma|^2
\]
Note that since both $s$ and $\sigma \in \R$ are scalars, then this potential can be defined globally as a simple quadratic error.
The introduction of this potential will bound growth of $\sigma(t)$ with respect to growth of $s(t)$, and vice-versa, ensuring that the two path parameters remain close during the system transient.

\begin{theorem}\label{th:synch}
Consider the combined dynamics \eqref{eq:q_dyn}
along with the reference system \eqref{eq:s_dyn}.
Assume there exists $\beta > 0$ such that
${q'_\text{r}}^\top M(q_\text{r}) q'_\text{r} > \beta$.
Define $\sigma(t)$ by \eqref{eq:sigma_dot} and \eqref{eq:sigma_integral}
and the path error by
\begin{align}
\tilde{s} := s - \sigma.
\label{eq:tilde_sigma}
\end{align}

Let $\tau_{\dot{q}}$ be given by \eqref{eq:tau_dot_q} and $\tau_\text{r}^0$ be given by \eqref{eq:tau_r^0}.
Let $\phi(q,q_\text{r})$ be a postive definite spring potential and define the interconnection control signals by
\begin{align}
  \tau &:= - \mu \frac{\tau_{\dot{q}}}{\sqrt{\pb{\tau_{\dot{q}}}{\dot{q}}}}
  + \tau_r^0 - \frac{\partial \phi(q,q_\text{r})}{\partial q},
  \label{eq:tau}  \\
  e_\text{c} & := -\kappa \tilde{s} - R \dot{\tilde{s}}
- \pb{\frac{\partial \phi(q,q_\text{r})}{\partial q_\text{r}}}{q'_\text{r}(s)},
\label{eq:e_c}
\end{align}
where $\kappa, R > 0$ are positive constants and
\begin{align}
\mu  & := \frac{1}{\sqrt{M_\text{r}(s(t))}} \left( \frac{\partial W_\text{r}}{\partial \sigma}  -\kappa \tilde{s} - R \dot{\tilde{s}} \right).
\label{eq:mu}
\end{align}
Then solutions of the closed-loop system converge to constant energy trajectories $\dot{E} \rightarrow 0$  with $\tilde{s}$ a constant: $\tilde{s}(t) \rightarrow \text{const}$.
\end{theorem}

\begin{remark}
Note that both $\tau_{\dot{q}} = \mathcal{O} (|\dot{q}|)$ and $\sqrt{\pb{\tau_{\dot{q}}}{\dot{q}}} = \mathcal{O}(|\dot{q}|)$ are asymptotically of order $|\dot{q}|$ for $\dot{q} \rightarrow 0$.
Since $M(q)$ is full rank positive definite it follows that $\tau$ \eqref{eq:tau} is well defined even for vanishing $\dot{q}$.
For $\dot{q} = 0$, I choose $\tau_{\dot{q}}/\sqrt{\pb{\tau_{\dot{q}}}{\dot{q}}} = 0$ even though the limit
$\lim_{\dot{q} \rightarrow 0} \tau_{\dot{q}}/\sqrt{\pb{\tau_{\dot{q}}}{\dot{q}}}$ may not be zero.
The resulting definition may be discontinuous at $\dot{q} = 0$.
Since the reference trajectory is chosen $\dot{q}_\text{r} \not= 0$ this discontinuity will not be present in the asymptotic solution of the system.
The transient response analysis of the system depends only on existence of the solutions but some care must be taken in the application of Barbalat's lemma.
\end{remark}

\begin{proof}
Consider the combined energy of the system
\begin{align}
E(q, \dot{q}, s, \dot{s}, \sigma)
& := H(q,\dot{q}) + H_\text{r} (s,\dot{s}) + W_\text{r}(\sigma)+ W_\text{r}(s) \notag \\
& \quad\quad  + \phi(q, q_\text{r}(s)) + \psi(s, \sigma)
\label{eq:E_interconned_system}
\end{align}
where $H(q,\dot{q})$ \eqref{eq:H} is the Hamiltonian for the true system,
$H_\text{r}(s,\dot{s})$ \eqref{eq:H_r} is the Hamiltonian for the reference system, the two terms $W_\text{r}(\sigma)$ and $W_\text{r}(s)$ are shaping potentials for the true and reference system, and $\phi$ and $\psi$ are the potentials associated with the spring forces introduced to synchronise the states.

Recall that the time derivative along trajectories of the Hamiltonians are
\begin{align*}
\ddt H(q,\dot{q}) = \pb{\tau}{\dot{q}},
\quad \ddt H_\text{r} (s,\dot{s}) = \pb{e_\text{c}}{\dot{s}}
\end{align*}
Substituting from \eqref{eq:tau} one has
\begin{align*}
\ddt H & =  \pb{-(\mu \tau_{\dot{q}} - \tau_r^0) - \frac{\partial \phi(q,q_\text{r})}{\partial q}}{\dot{q}} \\
& = - \mu \pb{\tau_{\dot{q}}}{\dot{q}} -
     \frac{\partial \phi(q,q_\text{r})}{\partial q}^\top \dot{q} \\
& =
- \frac{\pb{\tau_{\dot{q}}}{\dot{q}}}{\sqrt{\pb{\tau_{\dot{q}}}{\dot{q}}} \sqrt{M_\text{r}(s(t))}} \left( \frac{\partial W_\text{r}}{\partial \sigma}  - \kappa \tilde{s} - R \dot{\tilde{s}} \right)
\\
& \makebox[3cm]{}
- \frac{\partial \phi(q,q_\text{r})}{\partial q}^\top \dot{q} \\
& =
- \frac{\sqrt{\dot{q}^\top M(q)\dot{q} }}{\sqrt{M_\text{r}(s(t))}} \left( \frac{\partial W_\text{r}}{\partial \sigma}  - \kappa \tilde{s} - R \dot{\tilde{s}} \right)
- \frac{\partial \phi(q,q_\text{r})}{\partial q}^\top \dot{q} \\
& =
- \dot{\sigma} \left( \frac{\partial W_\text{r}}{\partial \sigma}  - \kappa \tilde{s} - R \dot{\tilde{s}} \right)
- \frac{\partial \phi(q,q_\text{r})}{\partial q}^\top \dot{q} \\
& =
- \ddt W_\text{r}(\sigma)  + \dot{\sigma}(\kappa \tilde{s} + R \dot{\tilde{s}})
- \frac{\partial \phi(q,q_\text{r})}{\partial q}^\top \dot{q} \\
\end{align*}
Similarly, computing $\ddt H_\text{r}$ and substituting from \eqref{eq:e_c} one has
\begin{align*}
\ddt H_\text{r}
& = \pb{ -\kappa \tilde{s} - R \dot{\tilde{s}}
- \pb{\frac{\partial \phi(q,q_\text{r})}{\partial q_\text{r}}}{q'_\text{r}(s)}}{\dot{s}} \\
& = -\dot{s} (\kappa \tilde{s} + R \dot{\tilde{s}})
- \frac{\partial \phi(q,q_\text{r})}{\partial q_\text{r}} \ddt q_\text{r}(s(t))
\end{align*}
Finally, differentiating $E$ and substituting for the relevant expressions one obtains
\begin{align*}
\ddt E & =
- \ddt W_\text{r}(\sigma)  + \dot{\sigma}(\kappa \tilde{s} + R \dot{\tilde{s}})
- \frac{\partial \phi(q,q_\text{r})}{\partial q}^\top \dot{q} \\
 & \quad -\dot{s} (\kappa \tilde{s} + R \dot{\tilde{s}})
- \frac{\partial \phi(q,q_\text{r})}{\partial q_\text{r}} \ddt q_\text{r}(s(t)) \\
 & \quad + \ddt W_\text{r} (\sigma) +
\frac{\partial \phi(q,q_\text{r})}{\partial q} \dot{q}
 +
\frac{\partial \phi(q,q_\text{r})}{\partial q_\text{r}} \ddt q_\text{r}(s(t)) \\
 & \quad +
\frac{\partial \psi(s,\sigma)}{\partial s}\dot{s} +
\frac{\partial \psi(s,\sigma)}{\partial \sigma}\dot{\sigma} \\
& =
  -( \dot{s} -\dot{\sigma})(\kappa \tilde{s} + R \dot{\tilde{s}}) \\
 &\quad - \frac{\partial \phi(q,q_\text{r})}{\partial q}^\top \dot{q}
- \frac{\partial \phi(q,q_\text{r})}{\partial q_\text{r}} \ddt q_\text{r}(s(t))  \\
 & \quad
+
\frac{\partial \phi(q,q_\text{r})}{\partial q} \dot{q} +
\frac{\partial \phi(q,q_\text{r})}{\partial q_\text{r}} \ddt q_\text{r}(s(t)) \\
 &\quad +
\kappa (s - \sigma) \dot{s} -
\kappa (s - \sigma) \dot{\sigma}  \\
& = - \dot{\tilde{s}}(\kappa \tilde{s} + R \dot{\tilde{s}})
+ \kappa \tilde{s}\dot{\tilde{s}} \\
& =  -  R \dot{\tilde{s}}^2
\end{align*}
It follows from LaSalles principle that $\dot{E}(t) \rightarrow 0$.
Furthermore, since the feedback interconnections are smooth and bounded except at points $\dot{q} = 0$, and these are isolated points in the trajectory, it follows that the derivatives of $\dot{\tilde{s}}$ are uniformly bounded except on a countable collection of times.
It follows that $\dot{\tilde{s}}$ is piecewise uniformly continuous and applying a generalisation of Barbalat's lemma \cite{2011_Wu_TAC} it follows that $\dot{\tilde{s}} \rightarrow 0$ except on a set of measure zero.
Since $\tilde{s}$ is continuous, then $\tilde{s}$ converges to a constant.
\end{proof}

The proposed synchronisation control is a closed passive interconnection of Hamiltonian systems and resulting combined system is itself a passive Hamiltonian system.
Theorem~\ref{th:synch} demonstrates that the combined solution will converge asymptotically to a trajectory that lies in the set of forward invariant trajectories characterised by $\dot{E} = 0$ and $\tilde{s} = \text{const.}$.
As shown in the next section, the free trajectories $\dot{s} = 1$ and $q(t) = q_\text{r}(t)$ certainly lie in this forward invariant set.
However, without additional control, the damping term will continue to extract energy from the system, even if it is in small increments due to small disturbances, until the limiting trajectory has zero kinetic energy and lies at a local minimum of the combined potential $U(q) + W_\text{r}(\sigma) + W_\text{r}(s) + \phi(q,q_\text{r}) + \psi(s,\sigma)$.
(Note that for $\dot{q} = 0$ then $\tau_{\text{r}}^0 = 0$ by construction.)
Such solutions also lie in (a separate component of) the forward invariant set of the combined system trajectories.
In particular, synchronisation of the trajectories by itself is not sufficient to ensure the desired trajectory tracking and in the next section I propose an additional external energy pumping control to stabilise the energy level the combined system.

\section{Energy Regulation}
\label{sec:energy_stabilisation}

This section presents two results; firstly an explicit demonstration that the reference trajectory is a forward invariant trajectory for the combined Hamiltonian system considered in Theorem \ref{th:synch}.
Secondly, I characterise the energy level of the reference trajectory for the combined system and propose a simple energy pump outer loop control that stabilises this energy level in the system and consequently stabilises the desired trajectory.

\begin{lemma}\label{lem:trajectory}
Let $t_0 > 0 \in \R$ be a real constant and consider $t > t_0$.
Consider the trajectory
\begin{align}
q_\star(t) & = q_\text{r}(t - t_0), \label{eq:q_is_qr}\\
s_\star(t) & = t - t_0 \label{eq:s_star_solution}
\end{align}
with additionally
\begin{align}
\sigma_\star(t_0) = s_\star(t_0). \label{eq:sigma_is_s}
\end{align}
Then $\sigma(t) = s(t)$ for all $t \geq t_0$ and $(q_\star(t), s_\star(t))$ is a forward invariant solution of the combined system in Theorem \ref{th:synch} with $\dot{E} = 0$ and $\tilde{s} = 0$.
\end{lemma}

\begin{proof}
Recalling \eqref{eq:sigma_dot} and using \eqref{eq:q_is_qr} and \eqref{eq:s_star_solution} one has that for all $t \geq t_0$
\begin{align*}
\dot{\sigma}(t) & = \frac{\sqrt{\dot{q}(t)^\top M(q(t)) \dot{q}(t)}}
{\sqrt{M_\text{r}(s(t))}} \\
 & = \dot{s}_\star(t) \frac{\sqrt{\dot{q}_{\text{r}}(t-t_0)^\top M(q_{\text{r}}(t-t_0)) \dot{q}_{\text{r}}(t-t_0)}}
{\sqrt{\dot{q}_{\text{r}}(t - t_0)^\top M(q_{\text{r}}(t - t_0) \dot{q}_{\text{r}}(t - t_0)}} \\
& = 1,
\end{align*}
since $\dot{s}_\star = 1$.
From \eqref{eq:sigma_is_s} it follows that $\sigma(t) = s(t)$ for all $t \geq t_0$ and both $\tilde{s} = 0$ and $\dot{\tilde{s}} = 0$ hold along the trajectory.
From \eqref{eq:q_is_qr} it follows that $\phi(q, q_\text{r}) = 0$ and hence
\[
\frac{\partial \phi}{\partial q}(q, q_\text{r}) = 0, \quad\quad\quad
\frac{\partial \phi}{\partial q_\text{r}}(q, q_\text{r}) = 0.
\]
Recalling \eqref{eq:e_c} it follows that $e_\text{c} = 0$ and it was shown that $s_\star(t)$ is a solution of the dynamics \eqref{eq:s_dyn} in \S\ref{sec:control}.
Using this in \eqref{eq:tau} one has
\[
\tau = - \mu \frac{\tau_{\dot{q}}}{\sqrt{\pb{\tau_{\dot{q}}}{\dot{q}}}}
  + \tau_r^0.
\]
Recalling \eqref{eq:tau_r^0} and solving for $\tau_{\dot{q}}$ one has
\begin{align*}
\tau_{\dot{q}}
&  =  \frac{M_\text{r}(s)}{\pb{\dot{q}}{\tau_r}}
\left( \frac{\pb{\tau_{\dot{q}}}{\dot{q}}}{M_\text{r}(s)} \tau_r -\tau^0_\text{r} \right)  =  \frac{M_\text{r}(s)}{\pb{\dot{q}}{\tau_r}}
\left( \tau_r -\tau^0_\text{r} \right)
\end{align*}
where $\pb{\tau_{\dot{q}}}{\dot{q}} = M_\text{r}(s)$ since $q = q_\text{r}$ and $\dot{s} = 1$.
The constant $\mu$ is given by
\begin{align*}
\mu  & = \frac{1}{\sqrt{M_\text{r}(s(t))}} \left( \frac{\partial W_\text{r}}{\partial \sigma} \right)
 = \frac{\dot{W}_\text{r}}{\sqrt{M_\text{r}(s(t))}} = - \frac{\pb{\tau_\text{r}}{\dot{q}_\text{r}}}{\sqrt{M_\text{r}(s(t))}}.
\end{align*}
where the final relationship follows from Lemma \ref{lem:dotW_is_poweralongref}.
Evaluating $\tau$ and substituting for $\mu$ and $\tau_{\dot{q}}$ one obtains
\begin{align*}
\tau & = \frac{\pb{\tau_\text{r}}{\dot{q}_\text{r}}}{\sqrt{M_\text{r}(s(t))}}
\frac{1}{\sqrt{\pb{\tau_{\dot{q}}}{\dot{q}}}}
\frac{M_\text{r}(s)}{\pb{\dot{q}}{\tau_r}}
\left( \tau_r -\tau^0_\text{r} \right)
+ \tau_\text{r}^0 \\
 & = \frac{1}{\sqrt{M_\text{r}(s)}}
\frac{M_\text{r}(s)}{\sqrt{\pb{\tau_{\dot{q}}}{\dot{q}}}}
\left( \tau_r -\tau^0_\text{r} \right)
+ \tau_\text{r}^0 \\
& =\left( \tau_r -\tau^0_\text{r} \right)
+ \tau_\text{r}^0 = \tau_r
\end{align*}
where the final line follows since $M_\text{r}(s) = \pb{\tau_{\dot{q}}}{\dot{q}}$ along trajectories $q = q_\text{r}$.
It follows from \eqref{eq:reference_trajectory} that $q_\star(t)$ is solution of \eqref{eq:q_dyn}.
This completes the proof.
\end{proof}

\begin{remark}
A key property of the proposed design is that the trajectory tracking can only be achieved up to an unknown constant time offset \eqref{eq:q_is_qr}.
This comes from the fundamental approach where the desired trajectory is embedded as zero dynamics in a larger system and then stabilised.
The philosophy of this approach draws strongly from the inspiration of Full \etal \cite{1999_Koditschek_JEB} and Westervelt \etal \cite{2003_Grizzle_2003} in stabilising zero dynamics of a larger system to obtain the desired behaviour as an emergent behaviour rather than forcing the desired response directly.
The consequence is that it is not the trajectory that is stabilised, rather it is a set of dynamics that generate the trajectory that is stabilised, and since these dynamics are time-invariant the best that can be achieved is stabilisation up to a time-offset.
\hfill$\Box$
\end{remark}

Although the trajectory $(q_\star, s_\star)$ is a free trajectory of the interconnected system, it is not the only such solution.
Understanding the full phase portrait of the interconnected system is far more complicated than those for the reference trajectory shown in Fig.~\ref{fig:phase_portrait} and beyond the scope of this paper (or the author).
However, the desired reference trajectory is closely associated with the shape of the phase portrait for the solution $\dot{s} = 1$ in the reference system, and a similar phase portrait for $\sigma$ that is `embedded' in the dynamics of the true system.
Consider the energy level \eqref{eq:E_interconned_system} of the interconnected system along the free trajectory $(q_\star, s_\star)$.
Both $\phi$ and $\psi$ potentials are zero, and $E_\text{r}$ \eqref{eq:E_r_zero} is also zero.
Finally, $H(q_\text{r}, \dot{q}_\text{r}) + W_\text{r}(\sigma)$ is also zero along this trajectory by the same argument used to derive \eqref{eq:E_r_zero}.
Thus, the desired trajectory is characterised by the zero energy level set of interconnected system.
Furthermore, it is straightforward to see that the energy of the true and reference systems separately are also zero.
From the phase portrait of the reference system (Fig.~\ref{fig:phase_portrait}) the only two zero energy solutions are characterised $\dot{s} = \pm1$.
Recalling \eqref{eq:sigma_dot} it is clear that $\dot{\sigma} >0$ and since $\tilde{s}$ converges to a constant, the solution $\dot{s} = -1$ cannot lie in the limit set of the closed-loop system.
Thus, trajectories $(q_\star, s_\star)$ (Lemma~\ref{lem:trajectory}), for some offset $t_0$, are the only forward invariant trajectories of the interconnected system that lie in the limit set of the closed-loop system with energy $E \equiv 0$.

This suggests a very simple outer level control whose goal is to stabilize the energy level of the interconnected system, while also stabilizing the internal energy of the true and reference systems to zero.
This approach is highly reminiscent of the energy pumping controls proposed in the late nineties \cite{1995_Spong_CSM}.
Since we have control in both the true system and the reference system then it is possible to inject energy in both systems at the same time and minimize the need for energy exchange during asymptotic tracking of the trajectory, as well as ensuring that the separate systems independently stabilise to zero energy level sets.

Choose a constant gain $k > 0$, then the proposed ``energy pump'' inputs are
\begin{align}
f := - k \left(H(q,\dot{q}) + W_\text{r}(\sigma)\right) \dot{q} \\
f_r := - k \left(H_\text{r}(s,\dot{s}) + W_\text{r}(s) \right) \dot{s}
\end{align}
where $f$ is applied to the true system and $f_r$ to the reference system.
The action of $f$ will inject energy
\[
\pb{f}{\dot{q}} = - k (H(q,\dot{q}) + W_\text{r}(\sigma)) \pb{\dot{q}}{\dot{q}}
\]
into the true system, acting to decrease the energy level when $H(q,\dot{q}) + W_\text{r}(\sigma)$ is positive and increase energy levels when it is negative.
Similarly, $f_r$ will stabilise the energy levels in the reference system to zero.
Both controls will become inactive along the desired trajectory $(q_\star, s_\star)$ and will only contribute infinitesimally to stabilise the energy level of the interconnected system.

\section{Graphical Representation}
\label{sec:bond_graph}

The full control can be graphically described in the bond graph shown in Figure \ref{fig:bondgraph}.
In this bond graph, the only external energy to the system is supplied (and removed) through the energy pump shown in the centre top of the figure and the stabilising dissipation $R \dot{\tilde{s}}^2$ term in the bottom right.
The remainder of the graph describes the interconnection of the true and reference systems.
The upper interconnection branch is the direct coupling of $q$ to $q_\text{r}$ through a spring potential $\phi$.
The lower left modulated transformer (MTF) and the potential $W_\text{r}(\sigma)$ is the coupling of the reference potential to the true system.
The spring-damper coupling of $s$ and $\sigma$ is shown in the bottom right of the diagram.

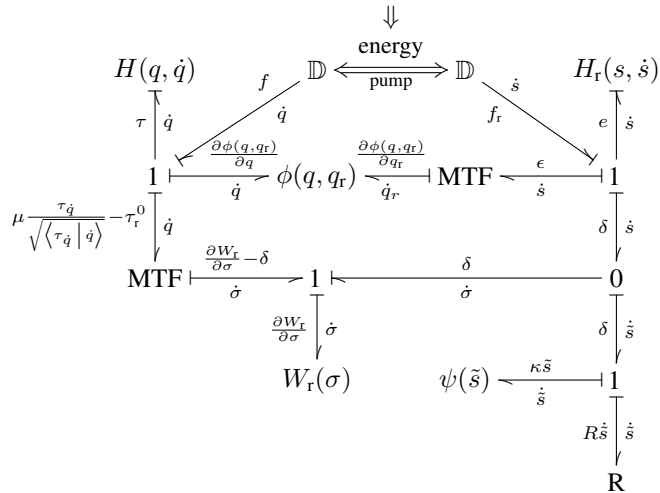
\begin{figure}[htb]
	\begin{centering}
		$
\xymatrix{ 
\\
    H(q,\dot{q})
	&
    \bgel{$\mathbb{D}$}
    \ar@{}[dl]^(0.9){}="b"_(0){}="a"
    \ar@^{->/|} "a":"b" ^(0.3){\dot{q}}_(0.3){f}
    \ar@{<=>}[r]^{\begin{array}{c} \Downarrow \\ \text{\small energy} \end{array}}_{\text{pump}}
	&
    \bgel{$\mathbb{D}$}
    \ar@{}[dr]^(0.9){}="b"_(0){}="a"
    \ar@^{->/|} "a":"b" ^(0.3){\dot{s}}_(0.3){f_\text{r}}
	&
	H_\text{r}(s,\dot{s})
\\
	\bgel{1}
	\ar@_{->/|}[u]_{\dot{q}}^{\tau}
\ar@^{/|->}[d]^{\dot{q}}_{\mu \frac{\tau_{\dot{q}}}{\sqrt{\pb{\tau_{\dot{q}}}{\dot{q}}}} - \tau_{\text{r}}^0}
	\ar@_{/|->}[r]_{\dot{q}}^(0.55){\frac{\partial \phi(q, q_\text{r})}{\partial q}}
	&
	\bgel{$\phi(q,q_\text{r})$}
	&
	\bgel{MTF}
	\ar@^{/|->}[l]^{\dot{q}_r}_{\frac{\partial \phi(q,q_\text{r})}{\partial q_\text{r}}}
	&
	\bgel{1}
	\ar@_{->/|}[u]_{\dot{s}}^{e}
	\ar@^{/|->}[l]^{\dot{s}}_{\epsilon}
	\ar@^{/|->}[d]^{\dot{s}}_{\delta}
\\
	\bgel{MTF}
	\ar@_{/|->}[r]_{\dot{\sigma}}^{\frac{\partial W_\text{r}}{\partial \sigma} - \delta}
	&
	\bgel{1}
	\ar@^{/|->}[d]^{\dot{\sigma}}_{\frac{\partial W_\text{r}}{\partial \sigma}}
	\ar@_{/|<-}[rr]_{\dot{\sigma}}^{\delta}
    &
	&
	\bgel{0}
	\ar@^{/|->}[d]^{\dot{\tilde{s}}}_{\delta}
	&
\\
    &		
    \bgel{$W_{\text{r}}(\sigma)$}
	&
	\bgel{$\psi(\tilde{s})$}
	&
	\bgel{1}
	\ar@^{/|->}[l]^{\dot{\tilde{s}}}_{\kappa\tilde{s}}
	\ar@^{/|->}[d]^{\dot{\tilde{s}}}_{R\dot{\tilde{s}}}
\\
    &		
	&
	&
	\bgel{R}
}$
\caption{Bond Graph representation of the proposed synchronisation control in Theorem \ref{th:synch}.}
\label{fig:bondgraph}
\end{centering}
\end{figure}

\section{Conclusion}

The proposed control architecture is fully non-linear and based on energy principles.
Due to its foundation, the approach is expected to provide highly robust control with strong passivity properties with respect to external disturbances.
There is considerable potential to consider implementations on non-standard state spaces such as $\SE(3)$ for mobile robotic vehicles as well as extensions based on the principles of passive systems \cite{2007_Secchi_book}.
Although the approach is derived for conservative mechanical systems it should be able to be applied to more liberal systems (with dissipative elements) as long as suitable conditions on the existence of $W_\text{r}$ are provided.

\section*{Acknowledgments}
This research was supported by the Australian Research Council
through the ``Australian Centre of Excellence for Robotic Vision'' CE140100016.



\end{document}